%% file: main.tex
\documentclass[11pt]{article}
\input{preamble}

\title{Deterministic Online Bipartite Edge Coloring}

\author[1]{Joakim Blikstad\thanks{Supported by the Swedish Research Council (Reg. No. 2019-05622) and the Google PhD Fellowship Program.}}
\author[2]{Ola Svensson\thanks{Supported by the Swiss State Secretariat for Education, Research and Innovation (SERI) under contract number MB22.00054.}}
\author[2]{Radu Vintan\protect\footnotemark[\value{footnote}]}
\author[3]{David Wajc\thanks{Supported by a Taub Family Foundation ``Leader in Science and Technology'' fellowship, and by ISF grant 3200/24.}}

\affil[1]{KTH Royal Institute of Technology \& Max Planck Institute for Informatics, \href{blikstad@kth.se}{blikstad@kth.se}}
\affil[2]{EPFL, \href{ola.svensson@epfl.ch}{ola.svensson@epfl.ch}}
\affil[3]{EPFL, \href{radu.vintan@epfl.ch}{radu.vintan@epfl.ch}}
\affil[4]{Technion --- Israel Institute of Technology,  \href{david.wajc@gmail.com}{david.wajc@gmail.com}}
\affil[ ]{\textit {}}

\date{\vspace{-1.3cm}}

\begin{document}
\maketitle

\pagenumbering{gobble}
\input{Sections/abstract}

\newpage
\pagenumbering{arabic}
\input{Sections/intro}
\input{Sections/prelims}
\input{Sections/adaptive}

\paragraph{Acknowledgements.} David Wajc would like to thank Neel Patel for useful discussions regarding contention resolution schemes, and thank Yuval Filmus for suggestions on presentation.

\bibliographystyle{alpha}
\bibliography{abb,ultimate,bibliography}

\end{document}

%% file: preamble.tex
\usepackage{inconsolata}
\usepackage{libertine}
\usepackage[margin=1in]{geometry}
\usepackage[utf8]{inputenc}
\usepackage{authblk}
\usepackage{amsmath, amssymb, amsthm, thmtools, amsfonts, bm, bbm, thm-restate}
\usepackage{algorithm}
\usepackage{algorithmicx}
\usepackage[noend]{algpseudocode}
\usepackage{cite}
\usepackage[numbers,sort]{natbib}

\usepackage{asymptote}
\usepackage{graphicx}
\usepackage{todonotes}
\usepackage{multirow}
\usepackage{comment}
\usepackage{dsfont}
\usepackage{bigstrut}
\usepackage{caption}
\usepackage{subcaption}
\usepackage{multirow}
\usepackage{makecell}
\usepackage{booktabs}
\usepackage{xcolor}
\definecolor{BrickRed}{rgb}{0.8,0.25,0.33}
\usepackage[pagebackref,citecolor=blue]{hyperref}

\usepackage{enumitem}
\usepackage[normalem]{ulem}

\usepackage[framemethod=tikz]{mdframed}
\usepackage{nicefrac}

\usepackage{tikz}
\usetikzlibrary{positioning, fit, calc}

\usepackage{pifont}

\usepackage[capitalize]{cleveref}

\declaretheorem[numberwithin=section,refname={Theorem,Theorems},Refname={Theorem,Theorems}]{theorem}

\declaretheorem[numberlike=theorem]{lemma}

\declaretheorem[numberlike=theorem]{corollary}
\declaretheorem[numberlike=theorem]{definition}

\declaretheorem[numberlike=theorem]{Remark}
\declaretheorem[numberlike=theorem,refname={Fact,Facts},Refname={Fact,Facts},name={Fact}]{fact}

\declaretheorem[numberlike=theorem, refname={Observation,Observations},Refname={Observation,Observations},name={Observation}]{observation}

\allowdisplaybreaks

\newcommand{\E}{\mathbb{E}}

\newcommand{\eps}{\varepsilon}

\newcommand{\calC}{\mathcal{C}}
\newcommand{\calD}{\mathcal{D}}

\newcommand{\poly}{\mathrm{poly}}

\newcommand{\Ber}{\mathrm{Ber}}

\let\vec\mathbf
\renewcommand{\vec}{\mathbf}

\usepackage{xspace}
\newcommand{\good}{$\eps$-good\xspace}
\newcommand{\bad}{$\eps$-bad\xspace}
\newcommand{\stepsize}{\frac{2}{\sqrt{\varepsilon}\Delta}}
\newcommand{\varsize}{\frac{2}{\varepsilon}}

%% file: Sections/abstract.tex
\begin{abstract}
    We study online bipartite edge coloring, with nodes on one side of the graph revealed sequentially.\\ The trivial greedy algorithm is $(2-o(1))$-competitive, which is optimal for graphs of low maximum degree, $\Delta=O(\log n)$ [BNMN IPL'92]. Numerous online edge-coloring algorithms outperforming the greedy algorithm in various settings were designed over the years (e.g., [AGKM~FOCS'03, BMM~SODA'10, CPW~FOCS'19, BGW~SODA'21, KLSST~STOC'22, BSVW~STOC'24]), all crucially relying on randomization. A commonly-held belief, first stated by [BNMN IPL'92], is that randomization is necessary to outperform greedy.

   \smallskip 
   
Surprisingly, we refute this belief, by presenting a \emph{deterministic} algorithm that beats greedy for sufficiently large $\Delta=\Omega(\log n)$, and in particular has competitive ratio $\frac{e}{e-1}+o(1)$ for all $\Delta=\omega(\log n)$. We obtain our result via a new and surprisingly simple \emph{randomized} algorithm that works against \emph{adaptive adversaries} (as opposed to oblivious adversaries assumed by prior work), which implies the existence of a similarly-competitive deterministic algorithm [BDBKTW STOC'90].
This is the first use of contention resolution schemes, which are randomized algorithms for randomized inputs, that yields a deterministic algorithm for deterministic settings.

\end{abstract}

%% file: Sections/intro.tex
\section{Introduction}

Consider a bipartite graph of maximum degree $\Delta$, with the nodes on one side revealed one after another.
An online algorithm must color arriving nodes' edges immediately and irrevocably on arrival, so that no two edges sharing an endpoint receive the same color.
(So, each color class is a matching in the graph.)
If the algorithm can color any such graph with $\alpha\Delta$ colors, then it is \emph{$\alpha$-competitive} with respect to the offline optimal solution, which requires only $\Delta$ colors, by K\"onig's line coloring theorem \cite{konig1916graphen}.

An early result of competitive analysis \cite{bar1992greedy} asserts that the naive $(2-1/\Delta)$-competitive greedy algorithm, which assigns each edge a lowest available color, is optimally competitive in the worst-case, specifically for bipartite graphs of low maximum degree $\Delta=O(\log n)$ under one-sided node arrivals \cite{cohen2019tight}. However, \cite{bar1992greedy} conjectured that better bounds are achievable for $\Delta=\omega(\log n)$, at least using randomization. To quote from their work:
\begin{quote}
\emph{... An interesting open problem is whether better bounds [than greedy's]
can be achieved for graphs whose maximal degree is larger [than $\log n$]. This seems less plausible in the
deterministic case, but perhaps one can devise a
randomized algorithm that would edge color a
graph better than the greedy algorithm for high
degree graphs.}
\end{quote}
And indeed, a long line of work \cite{aggarwal2003switch,bahmani2012online,cohen2019tight,bhattacharya2021online,saberi2021greedy,kulkarni2022online,naor2025online,blikstad2023simple,blikstad2024online} made progress on the \cite{bar1992greedy} conjecture, culminating in a $(1+o(1))$-competitive algorithm for general graphs with $\Delta=\omega(\log n)$ under edge arrivals \cite{blikstad2024online}.
Crucially, all above works' algorithms are randomized. 
Indeed, we recall that \cite{bar1992greedy} were skeptical that deterministic algorithms may outperform the greedy algorithm.

Generally, deterministic algorithms have large gaps compared to randomized algorithms for many online problems, most notably exponential gaps for $k$-server \cite{koutsoupias1995k,bansal2015polylogarithmic} and for caching specifically \cite{manasse1988competitive,fiat1991competitive}.
Indeed, deterministic online caching algorithms cannot provide better worst-case guarantees than trivial algorithms (see discussion in \cite[Chapter 24]{roughgarden2021beyond}).
This is also the case for the ``dual'' problem to ours, of online bipartite matching (packing a large matching, rather than covering the graph with few matchings) under one-sided node arrivals, randomization is necessary to beat the naive greedy algorithm \cite{karp1990optimal}. 
The above justifies the aforementioned belief that greedy is optimal among deterministic online edge-coloring algorithms.

We show that randomization is not necessary to outperform greedy in bipartite graphs under one-sided node arrivals, as studied in randomized settings by \cite{cohen2019tight}.

\begin{theorem}
\label{thm:main}
    There exists a deterministic $\big(\frac{e}{e-1}+o(1)\big)$-competitive online bipartite edge-coloring algorithm under one-sided node arrivals for bipartite graphs with known $\Delta=\omega(\log n)$.\footnote{Without knowing $\Delta$, a competitive ratio of $(\frac{e}{e-1}+o(1))$ is optimal for \emph{randomized} algorithms under such arrivals~\cite{cohen2019tight}.}
\end{theorem}

\subsection{Overview of our approach}

As mentioned, all prior online edge-coloring algorithms except greedy are randomized.
Moreover, these algorithms all assume an \emph{oblivious adversary},  that creates the input in advance without seeing the algorithm's choices.
In contrast, a more challenging \emph{adaptive adversary} creates the input based on the algorithm's choices and randomness so far. Unfortunately, against such an adversary randomization offers no advantage over deterministic algorithms  \cite{ben1994power} 
(see \Cref{lem:adaptive-reduction}). This suggests an approach to prove the existence of  deterministic algorithms: we should design \emph{randomized} algorithms, but ones that are competitive against adaptive adversaries!

\paragraph{The power of the adversary.} The key challenge with this approach (and reason to doubt the existence of deterministic algorithms) is the adaptive adversary's huge number of choices. 
To illustrate this, note that against an oblivious adversary, that fixes a graph in advance, an upper bound of $1/\poly(n)$ on the probability of some bad events for any possible arriving neighborhood suffices to union bound over all arrivals' possible bad events and show that it is unlikely for any to occur.
In contrast, against an adaptive adversary we require significantly sharper bounds, even for a single time step, where an adaptive adversary has up to ${n\choose \Delta}$  choices of $\Delta$ offline neighbors of the arriving online node.\footnote{A similar phenomenon occurs in \cite{kulkarni2022online}, who subsample graphs of high girth $\ell$. Key to their analysis is that each edge belongs to only $\Delta^{\ell}$ many cycles of girth less than $\ell$ in an oblivious setting. For an adaptive setting, this bound becomes $n^{\ell}$, and so union bounding over all such cycles to guarantee high girth requires concentration $1/n^{\ell}$ and not $1/\Delta^{\ell}$.}
This problem is further compounded by factoring in the number of future choices  of the entire graph available to the adversary, which may be super-exponential \cite{mckay1984asymptotics}.

\paragraph{Overcoming the adversary.}
We provide an algorithmic approach which allows us to focus on only roughly ${n \choose \Delta}$ bad events, which we show have exceedingly low probability, $\exp(-\Theta(\Delta^2))\leq 1/{n\choose \Delta}$ (using $\Delta = \Omega(\log n)$), even against an adaptive adversary. This allows us to union bound over these bad events, rather than over all possible future choices.
In contrast, all previous randomized algorithms had only $\exp(-\Delta)$ bad event probability, which is sufficient for the polynomially-many events presented by an oblivious adversary, but not for the $n^\Delta$ many choices of the adaptive one.

Concretely, our algorithm takes the following approach.
Initially, we assign each offline node the same palette $[\Delta+o(\Delta)]$.
For each online node $v_t$, we have each edge $(u,v_t)$ select a color $c$ uniformly and independently at random among the colors still available to $u$, i.e., this edge selects $c$ with probability $$x^{(t)}_{uc}:=\frac{\mathds{1}[c \textrm{ still available for }u]}{|\{c' \textrm{ still available for } u\}|}.$$
Our plan is to assign each color only to a single edge that selected it, and (for now) assign no colors to the other edges. To guarantee edges a good probability of being colored, we therefore need to control the number of collisions per color.
Versus an oblivious adversary, the above random choices result in each edge $(u,v_t)$ selecting each color $c\in [\Delta+o(\Delta)]$ uniformly, i.e., with marginal probability $1/(\Delta+o(\Delta))$, and so each color $c$ is selected less than once in expectation by the $\Delta$ neighbors of $v_t$, i.e., $\E[\sum_u x^{(t)}_{uc}] = 1-o(1)$.\footnote{\cite{blikstad2023simple} intuitively follow a similar approach, but use correlated choices for different nodes and colors, by sampling a single matching between nodes and colors. Unfortunately, an adaptive adversary can use this correlation to break their algorithm.}
Unfortunately, an \emph{adaptive adversary} can easily make the expected number of times a color is selected much larger.\footnote{This is problematic for the common approach used for online edge coloring against an oblivious adversary, by iteratively computing ``fair'' matchings online, corresponding to the different colors \cite{cohen2019tight,saberi2021greedy,kulkarni2022online,blikstad2024online}.} The hope is to show that the adversary cannot accomplish this for too many colors.

Our key idea, allowing us to obtain $\exp(-\Delta^2)$ type concentration, is to focus on \emph{sets} of some $\varepsilon\cdot \Delta$ colors $C$, and sets of $\Delta$ offline nodes $U\subseteq V$, corresponding to potential neighborhoods of online nodes. 
We note that if $U$ is the next online node's neighborhood, then if on average the colors in $C$ have low load, i.e., $\sum_{u\in U}\sum_{c\in C}x^{(t)}_{uc} \leq (1+\varepsilon)\cdot |C|$, and this holds for all pairs $(U,C)$, then the fraction of colors whose load may exceed $1+\varepsilon$ is at most $\varepsilon$ (\Cref{obs:no-bad-implies-mostly-good}).
Thus, we reduce the number of bad events we care about to a moderate ${n\choose \Delta}\cdot \exp(\Delta) = n^{O(\Delta)}$, one for each pair $(U,C)$, as opposed to needing to concern ourselves with all possible futures.
The key challenge is to prove that any pair's average load exceeding $(1+\eps)$ occurs with probability $1/n^{\Omega(\Delta)}$, despite the adversary's adaptive choices.

\paragraph{Proving concentration via martingales.}
To prove strong concentration of $(U,C)$ pairs' average load, for each pair $(U,C)$, we set up a martingale $Z_0,Z_1,\dots,Z_m$, with $Z_i$ taking on the value $$\sum_{u\in U}\sum_{c\in C} \frac{\mathds{1}[c \textrm{ still available for }u]}{|\{c' \textrm{ still available for } u\}|}$$ after the $i^{th}$ edge picked a random color and removed it from its palette.
Crucially, we show that these are indeed martingales, even when facing an adaptive adversary, and that they have constant variance and $O(1/\Delta)$ maximum step size.
Applying Freedman's inequality (\Cref{thm:freedman_inequality}) then allows us to prove that the average load of any fixed pair $(U,C)$ exceeds its expectation of $|C|$ by $\varepsilon|C|=\eps^2|\Delta|$ occurs with probability at most $\exp(-\poly(\eps)\cdot \Delta^2)$.
For $\Delta=\Omega(\log n)$ sufficiently large, this is $1/n^{\Omega(\Delta)}$, which allows us to union bound over all $n^{O(\Delta)}$ bad events and show that with high probability, none occur.\footnote{We remark that the recent $(1+o(1))$-competitive edge coloring algorithm of \cite{blikstad2024online} (under oblivious adversary and edge arrivals) also relied crucially on martingales, though for a very different reason. They did so to control problematic correlations between nodes, while we use martingales primarily to avoid union bounding over all potential futures/histories.}

\paragraph{Contention Resolution.}
So far, we outlined how we prove that with high probability, most colors $c$ have load $\sum_u x^{(t)}_{uc}\leq 1+o(1)$. 
This does not result in each edge being colored, though, as verified for the case $x^{(t)}_{uc} \approx \frac{1}{\Delta}$ for each of the $\Delta$ nodes $u\in N(v_t)$ and each color $c$, where the probability a color is selected is only $1-(1-1/\Delta)^\Delta\approx 1-e^{-1}$, and so at most $\approx \Delta(1-e^{-1})$ colors are assigned (and edge colored).  
This turns out to be the worst case, and a ratio of $1-e^{-1}$ is achievable:
To guarantee each edge is colored with probability close to $1-e^{-1}$, we extend the current coloring by assigning each color $c$ to one of the edges that selected it using a \emph{contention resolution scheme} \cite{feige2006maximizing} (see \Cref{lem:CRS}). 
This results in us assigning color $c$ to edge $(u,v_t)$ with probability roughly $(1-e^{-1})\cdot x^{(t)}_{uc}$ (\Cref{lem:per-edge-coloring-prob}).
Thus, we color each edge with probability $1-e^{-1}$, and by simple concentration inequalities all high-degree nodes have their degree decrease by a factor of roughly $e$ with high probability.

\paragraph{Recursing.}
Our approach decreases the maximum degree of the uncolored subgraph by a factor of roughly $e$ with high probability~using only $\Delta(1+o(1))$ many colors. 
We invoke multiple copies of this algorithm (interleaved in an online fashion),  using a total of
$\Delta(1+o(1))\cdot (1+e^{-1}+e^{-2}+\dots) = \Delta\big(\frac{e}{e-1}+o(1)\big)$ many colors. This way, we decrease the uncolored subgraph's maximum degree to $o(\Delta)$, which we then color  
greedily using a further $2\cdot o(\Delta)=o(\Delta)$ more colors.
Overall, this algorithm requires $\Delta\big(\frac{e}{e-1}+o(1)\big)$ colors with high probability against an adaptive adversary. This, by \Cref{lem:adaptive-reduction}, yields a deterministic $\big(\frac{e}{e-1}+o(1)\big)$-competitive online edge-coloring algorithm, as stated in \Cref{thm:main}.

Our randomized algorithm and its analysis are presented in their entirety in \Cref{sec:adaptive}.

\subsection{Related work}

The first online edge-coloring algorithms (beyond greedy) worked under random-order edge arrivals, where a worst-case graph is revealed in random order. \cite{aggarwal2003switch} gave a $(1+o(1))$-competitive algorithm for multigraphs with $\Delta=\omega(n^2)$, while \cite{bahmani2012online} showed that for simple graphs with only $\Delta=\omega(\log n)$, one can achieve a competitive ratio of $1.27$, later improved to $1+o(1)$ \cite{bhattacharya2021online,kulkarni2022online}.
The first work to obtain results under (oblivious) adversarial arrivals was \cite{cohen2019tight}, who showed that $(1+o(1))$-competitiveness is achievable for known $\Delta=\omega(\log n)$ and that without knowledge of $\Delta$, the optimal competitive ratio is $\frac{e}{e-1}+o(1)$. The former result of \cite{cohen2019tight} was later simplified by \cite{blikstad2023simple}.
Greedy was subsequently beaten in general graphs \cite{saberi2021greedy}, even under edge arrivals, for which \cite{kulkarni2022online} gave a competitive ratio of $\frac{e}{e-1}+o(1)$, very recently improved to $1+o(1)$ \cite{blikstad2024online}.  As stated before, all these results are randomized against an oblivious adversary.

Contention resolution schemes (CRS), which we use, provide a uniform approach for rounding linear programs with packing constants: 
round elements independently, and then resolve contention in a way that guarantees both (i) feasibility of the output set, and (ii) that each rounded element is in the output set with good probability.
This approach was introduced in the context of welfare maximization, where the packing constraint is a rank-one matroid  \cite{feige2006maximizing,feige2006allocation}. This was later formalized as a more general rounding approach for matroid constraints and beyond, by \cite{chekuri2014submodular}.
Contention resolutions schemes have wide-ranging applications, including in submodular maximization \cite{feldman2011unified,chekuri2014submodular}, stochastic probing problems \cite{gupta2013stochastic}, combinatorial sparsification \cite{dughmi2022sparsification}, and recently in online algorithms (against oblivious or even stochastic oblivious adversaries) \cite{naor2025online,patel2024combinatorial}. 
They have also been extended to online settings \cite{feldman2016online}, which have found even more applications, particularly to stochastic problems, most prominently prophet inequalities.
Despite the ubiquity of contention resolution schemes, to the best of our knowledge, this paper presents the first example of a (randomized) CRS being used to design a deterministic algorithm for a deterministic problem.

\paragraph{Concurrent work.}
Concurrently, \cite{dudeja2025randomized} obtain orthogonal results for online edge coloring, via a different algorithm and analysis from ours. They show how to obtain deterministic $(1+\varepsilon)$-competitive online edge coloring under edge arrivals for any constant $\varepsilon>0$, in dense graphs with $\Delta=\Omega(n)$, for sufficiently large $n$.
In contrast, we improve on the greedy algorithm for one-sided node arrivals in bipartite graphs for all sufficiently large $\Delta=\Omega(\log n)$, thus matching the range of degrees of the lower bound of~\cite{bar1992greedy}.

%% file: Sections/prelims.tex
\section{Preliminaries} \label{sec:prelims}

\textbf{Problem definition and notation.} 
In online bipartite edge coloring, the input is an unknown bipartite graph of maximum degree $\Delta$ (this value is known in advance), with $n$ \emph{offline nodes} forming one side of the bipartition, and known up front; at each time $t=1,\dots,n$, a \emph{online node} $v_t$ on the other side is revealed, together with its edges to its neighbors, $N(v_t)$.
An online algorithm must decide for each edge on its arrival what color to assign it, immediately and irrevocably.
The subgraphs induced by each color must form a matching; put otherwise, each node must have at most one edge of each color.
The minimal number of colors (matchings) needed to cover such a graph is $\Delta$, by K\"onig's line coloring theorem \cite{konig1916graphen}. We say an online algorithm is \emph{$\alpha$-competitive} if it uses at most $\alpha$ times more colors than achievable offline, i.e., if the algorithm uses $\alpha\Delta$ colors. Our focus is on randomized algorithms that are $\alpha$-competitive w.h.p.~(with high probability, $1-1/n^c$ for constant $c>1$), which immediately imply algorithms that are $(\alpha+o(1))$-competitive in expectation (i.e., use $(\alpha+o(1))\Delta$ colors in expectation), by simply running greedy with new colors in the unlikely event that the high-probability algorithm uses more than $\alpha\Delta$ colors.

\paragraph{Adversaries.}
We focus on randomized algorithms whose input (online nodes and their edges) is generated \emph{adaptively}, based on the algorithm's prior random choices, by a so-called \emph{adaptive adversary}.\footnote{\cite{ben1994power} further distinguish between offline-adaptive and the weaker online-adaptive model. Our algorithm succeeds versus the stronger adversary, and since we do not consider the weaker adversary, we omit this distinction.} 
A classic result asserts that against such adversaries, randomness yields no benefit.

\begin{lemma}[\!\!\cite{ben1994power}] \label{thm:adaptive_implies_deterministic}\label{lem:adaptive-reduction}
    If there exists a $\gamma$-competitive randomized online edge-coloring algorithm against adaptive adversaries,
    then there exists a deterministic $\gamma$-competitive online edge-coloring algorithm.
\end{lemma}
\begin{proof}(Sketch.) The proof of the above lemma is direct, though generally yields computationally-inefficient algorithms: traverse the expectimax game tree for the game played between algorithm (minimizing expected competitive ratio) and adversary (maximizing competitive ratio). At  each algorithm-node of the tree, proceed to the child whose expectimax value is as good as that of the current node. See \cite[Chapter 7]{borodin2005online} for a more in-depth discussion of this proof for request-answer games more broadly.
\end{proof}

We note that, as outlined previously, our analysis relies on showing concentration of $n^{O(\Delta)}$ many martingales. By the method of pessimistic estimators \cite{raghavan1988probabilistic}, this allows us to derandomize the algorithm slightly more efficiently with a running time of $n^{O(\Delta)}$, which is sub-exponential for $\Delta=o(n/\log n)$. We leave it as an open problem to get polytime deterministic online edge-coloring algorithms beating greedy.

\paragraph{Contention Resolution Schemes.}
Our algorithms rely on single-item
\emph{contention resolution schemes} (CRSes), due to \cite{feige2006allocation}.
\begin{lemma}[\!\!\cite{feige2006allocation}]\label{lem:CRS}
     Let $\calD$ be a product distribution over subsets of $[n]$, with marginals $x_i$ for $i\in [n]$.
     Then, there exists a randomized algorithm CRS($R,\vec{x}$) which on input set $R\sim \calD$  outputs a subset $O \subseteq R$ that when $R$ is non-empty contains a single element, i.e., $|O|=1$, satisfying
     $$\Pr[i\in O] \geq x_i \cdot \frac{1 - \prod_{j \in [n]} (1 - x_j)}{\sum_{j\in [n]} x_j} \qquad\qquad  \forall i\in [n].$$
\end{lemma}
Strictly speaking, the guarantee that $O$ be non-empty when $R$ is non-empty is not usually provided in definitions of CRSs, but it is easy to satisfy by picking an arbitrary element in $R$ if none is selected.

\paragraph{Martingales.}
Our analysis relies on Freedman's inequality for martingales, whose definition we now recall.
\begin{definition}
A sequence of random variables $Z_0,Z_1,\dots,Z_m$ is a  \emph{martingale} if 
$$\E[Z_i \mid Z_0,Z_1,\dots,Z_{i-1}] = Z_{i-1} \qquad \forall i\in [m].$$
\end{definition}

\begin{lemma}[Freedman's Inequality \cite{freedman1975tail}]\label{thm:freedman_inequality}
    Let $Z_0,\dots,Z_m$ be a martingale. 
    If $|Z_i - Z_{i-1}| \leq A$ for all $i \geq 1$ and
    $$\sum_{i=1}^{m} \E[ (Z_i - Z_{i-1})^2 \mid Z_0,Z_1,\dots,Z_{i-1}] \leq \sigma^2$$ always.
    Then, for any real $\lambda \geq 0$:
    \begin{equation*}
        \Pr[ Z_m - Z_0 \geq \lambda ] \leq \exp\left( - \frac{\lambda^2}{2(\sigma^2 + A \lambda / 3)} \right).
    \end{equation*}
\end{lemma}

Finally, we need the following simple inequality.
\begin{fact} \label{lemma:aux_inequality}
    Let $0 \leq x \leq 1$. Then 
        $\frac{1 - \exp(-1-x)}{1+x} \geq 1 - e^{-1} - x.$
\end{fact}
\begin{proof}
    First, for $x \leq 1$ we have that $\exp(-x) = \sum_{i\geq 0} \frac{(-x)^i}{i!} \leq 1-x+\frac{x^2}{2} \leq 1 - \frac{x}{2}$. Furthermore, for any $x \geq 0$ we have: $\frac{1}{1 + x} \geq 1 - x$. Therefore, we obtain:
    \begin{align*}
        \frac{1 - \exp(-1-x)}{1+x} & \geq \left(1-e^{-1}\cdot \left(1-\frac{x}{2}\right)\right)\cdot (1-x) = 1-e^{-1} - x + \frac{3x}{2e} - \frac{x^2}{2e} \geq 1-e^{-1} - x,
    \end{align*}
    which ends the proof.
\end{proof}

%% file: Sections/adaptive.tex
\section{An Adaptive Algorithm}\label{sec:adaptive}

In this section we present our $\big(\frac{e}{e-1}+o(1)\big)$-competitive online edge-coloring algorithm that works against adaptive adversaries. Our main algorithm is an algorithm that decreases the uncolored subgraph's maximum degree at a rate consistent with the above.

For $\eps := 2\sqrt[5]{({\ln n})/{\Delta}}$, our main algorithm provides a $\Delta(1+\sqrt{\eps})$-edge-coloring of a (large) subgraph of $G$ of maximum degree $\Delta(G)\leq \Delta$, as follows.
Initially, we provide each node with the same \emph{palette} of available colors, $P(u)\gets \lceil (1+\sqrt{\eps})\Delta \rceil$.
At each time step we have each edge $(u,v_t)$ pick a single color $c\in P(u)$ uniformly at random, i.e., with probability $x^{(t)}_{uc} \gets \frac{1}{|P(u)|}$, remove this color from $P(u)$ (even if we do not finally assign this color to edge $(u,v_t)$). We then use a contention resolution scheme (\Cref{lem:CRS}) for each color $c$ to pick a single edge $(u,v_t)$ that selected $c$. 
Our pseudocode is given in \Cref{alg:adaptive-edge-coloring}.

\begin{algorithm}[!htb]
	\caption{Partial Edge Coloring}
	\label{alg:adaptive-edge-coloring}
	\begin{algorithmic}[1]
        \State \textbf{Input parameter:} $\Delta\in \mathbb{Z}$. \Comment{Intuitively, $\Delta(G)\leq \Delta$.}
        \State \textbf{\underline{Initialization:}} Set $\eps\gets 2\sqrt[5]{\frac{\ln n}{\Delta}}$. \textbf{For each} offline node $u$: create palette $P(u)\gets \mathcal{C} := \lceil (1+\sqrt{\varepsilon})\Delta \rceil$.
		\For{\textbf{each} online node $v_t$ on arrival}
			\For{\textbf{each} color $c \in \mathcal{C}$ and node $u \in N(v_t)$} 
			\State Set $x^{(t)}_{uc} \gets \frac{\mathds{1}[c\in P(u)]}{|P(u)|}$.
			\EndFor 
			\For{\textbf{each} edge $e=(u,v_t)$}
			\State Pick color $c(e) \in P(u)$ uniformly at random. \Comment{$\Pr[c(e)\textrm{ picked}] = x_{uc}^{(t)}$}
			\State Remove $c(e)$ from $P(u)$.
			\EndFor
			\For{\textbf{each} color $c \in \mathcal{C}$}
			\State Set $R_c \gets \{e = (u,v_t) \mid c(e) = c\}$.
			\State Set $\overrightarrow{x_c}$ to be the vector $(x^{(t)}_{uc})_{u \in N(v_t)}$.
			\State Set $S_c \gets  CRS(R_c,\overrightarrow{x_c})$.
			\If{$S_c\neq \emptyset$}
            \State Assign edge $e_c\in S_c$ the color $c$.
            \EndIf
			\EndFor
		\EndFor
	\end{algorithmic}	
\end{algorithm}	

As we show, this algorithm decreases the maximum degree of the uncolored subgraph in a dependable rate (i.e., w.h.p.), while only using $\frac{e}{e-1}+o(1)$ times more colors. Our main technical result is the following:

\begin{restatable}{theorem}{adaptivesubroutine}\label{thm:adaptive-subroutine}
    \Cref{alg:adaptive-edge-coloring} applied to bipartite graphs $G$ of maximum degree $\Delta(G)\leq \Delta$ for $\Delta \geq 32\cdot \ln n$ yields a feasible $\Delta(1+\sqrt{\varepsilon})$-edge-coloring of a subgraph $H\subseteq G$ such that $\Delta(G\setminus H)\leq (e^{-1} + 3\sqrt{\eps})\cdot \Delta$ with probability at least $1-n^{-5}$, and this guarantee holds even against an adaptive adversary.
\end{restatable}

Following \Cref{alg:adaptive-edge-coloring} by the greedy algorithm using new colors numbered $(1+\sqrt{\varepsilon})\Delta+1$ and higher (or more precisely, interleaving the two algorithms) directly implies a $(1+\sqrt{\varepsilon}+2e^{-1}+6\eps)\approx 1.73$-competitive online edge-coloring algorithm versus adaptive adversaries. However, we can do better by pipelining a number of invocations of \Cref{alg:adaptive-edge-coloring} in an online fashion; concretely, all edges of an arriving online vertex not colored by the $i^{th}$ copy of algorithm \Cref{alg:adaptive-edge-coloring}, are then revealed in an online fashion to the $(i+1)^{th}$ copy. Each copy of \Cref{alg:adaptive-edge-coloring} roughly decreases the degree by a factor of $e$ (w.h.p.). 
Repeating this until the residual degree is $o(\Delta)$, 
while using roughly $\Delta(1+e^{-1}+e^{-2}+\dots) \approx 1.58 \Delta$ colors, we can do better, as in the following theorem.

\begin{restatable}{theorem}{adaptivealgo}\label{thm:adaptive-algo}
    There exists a randomized online bipartite edge-coloring algorithm that is $\big(\frac{e}{e-1}+100\sqrt[11]{\frac{\ln n }{\Delta}}\big)$-competitive w.h.p.~against an adaptive adversary, provided $\Delta\geq 10^{11}\cdot \ln n$. In particular, for $\Delta=\omega(\log n)$, this algorithm's competitive ratio is $\big(\frac{e}{e-1}+o(1)\big)$.
\end{restatable}

The proof is fairly direct, though somewhat calculation heavy due to quantification of the $o(1)$ terms, and is therefore deferred to \Cref{sec:partial-to-full}.

By combining \Cref{thm:adaptive-algo} with \Cref{lem:adaptive-reduction} we obtain our main result.
\begin{theorem}\label{thm:deterministic-algo}
    There exists a deterministic online bipartite edge-coloring algorithm with competitive ratio less than two for for bipartite graphs with sufficiently large maximum degree $\Delta=\Omega(\log n)$ and competitive ratio $\big(\frac{e}{e-1}+o(1)\big)$ for $\Delta = \omega(\log n)$.
\end{theorem}

The next section is dedicated to the core of our analysis, namely proving \Cref{thm:adaptive-subroutine}, whereby \Cref{alg:adaptive-edge-coloring} provides a partial edge coloring that decreases the maximum degree of the uncolored graph at a rate of roughly $1-e^{-1}$ per color, w.h.p., even against adaptive adversaries. 

\subsection{Proof that \Cref{alg:adaptive-edge-coloring} Provides an Effective Partial Coloring}

Note that by construction, for any time $t$ and node $u\in N(v_t)$, we have that $\sum_{c} x^{(t)}_{uc}=1$. We show that if a color $c$ (nearly) satisfies its analogous fractional degree constraints, this yields a high probability of coloring edges $(u,v_t)$ that selected color $c$.
\begin{definition}
We say a color $c$ is \emph{\good at time $t$} if $\sum_{u\in N(v_t)} x^{(t)}_{uc}\leq 1+\varepsilon$.
\end{definition}
\begin{lemma}\label{lem:per-edge-coloring-prob}
    For fixed time $t$ and values $x^{(t)}_{uc}$, if color $c$ is \good at time $t$ and $u\in N(v_t)$, then:
    $$\Pr[(u,v_t) \mathrm{\ colored\ } c] \geq x^{(t)}_{uc}\cdot (1-e^{-1}-\varepsilon).$$
\end{lemma}
\begin{proof}
For brevity, we use $x_{uc}$ and $N$ as shorthand for $x^{(t)}_{uc}$ and $N(v_t)$.
Since edge $(u,v_t)$ selects color $c$ at time $t$ independently with probability $x_{uc}$,
by the properties of CRS (\Cref{lem:CRS}) we have that
\begin{align*}\frac{\Pr[(u,v_t) \mathrm{\ colored\ } c]}{x_{uc}} 
\geq \frac{1-\prod_{u\in N} (1-x_{uc})}{\sum_{u\in N} x_{uc}}
\geq  \frac{1-\exp\left(- \sum_{u\in N} x_{uc}\right)}{\sum_{u\in N} x_{uc}} \geq  \frac{1-\exp(-1-\eps)}{1+\eps} \geq (1-e^{-1}-\varepsilon),
\end{align*}
where the second inequality follows from $1-z\leq \exp(-z)$ for real $z$, and the last two inequalities follow from $\sum_{u\in N} x_{uc} \leq 1+\eps$ and monotonicity of $f(z):=(1-\exp(-z))/z$, and by \Cref{lemma:aux_inequality}, respectively.
\end{proof}

If all colors are \good at time $t$, then, since offline nodes' fractional degree is equal to one, $\sum_c x^{(t)}_{uc}=1$, this implies that each edge $(u,v_t)$ is colored with probability close to $1-e^{-1}$, which intuitively (and formally proven in more general settings in \Cref{lem:degree-rate-decrease}) allows us to argue that the maximum degree decreases at a similar rate w.h.p.
Unfortunately, it is quite likely for \emph{some} colors to violate this fractional degree constraint and thus not be \good, at least for \emph{some} adaptive choice of $N(v_t)$.
We will therefore attempt to prove that with high probability (even against an adaptive adversary), \emph{most} colors are \good.

To prove that most colors are \good, i.e., nearly satisfy the fractional degree constraint, we show that no large set of colors $C$ violates this constraint on average, for any set of neighbors $U$ of $v_t$ possibly selected by the adversary. (This allows us to appeal to large deviation inequalities for a number of martingales that we set up shortly.)
The following definition captures this bad event that we wish to avoid. 
\begin{definition} 
We say a pair $(U,C)$ consisting of $\Delta$ offline nodes $U$ and of $\varepsilon \Delta$ colors $C \subseteq \mathcal{C}$ is \emph{\bad at time $t$} if $\sum_{u\in U} \sum_{c\in C} x^{(t)}_{uc} > (1 + \eps)\cdot |C|$. 
\end{definition}
A simple argument implies that if no pairs $(U,C)$ are \bad at time $t$ and $N(v_t)\subseteq U$, then most colors are \good at time $t$.
\begin{observation}\label{obs:no-bad-implies-mostly-good}
    Fix a set $U$ of $\Delta$ offline nodes. If no pair $(U,C)$ is \bad at time $t$, then if $N(v_t)\subseteq U$, then there are at most $\varepsilon \Delta$ colors $c$ which are not \good at time $t$.
\end{observation}
\begin{proof}
    Assume by way of contradiction that some set $C'\subseteq \calC$ of $\varepsilon\Delta$ colors $c$ are not \good. Then 
    $$\sum_{c\in C'} \sum_{u\in U} x^{(t)}_{uc} \geq \sum_{c\in C'} \sum_{u\in N(v_t)} x^{(t)}_{uc}  > (1+\varepsilon)\cdot |C'|,$$
    where the first inequality follows from $N(v_t) \subseteq U$ and the latter follows from our assumption regarding $C'$. But then $(U,C')$ is bad at time $t$, and we obtain our desired contradiction.
\end{proof}

The above, together with a simple counting argument below, implies that in order to rule out the existence of many colors that are not \good, we do not need to union bound over all futures, but can union bound over a (quite modest) number of \bad pairs.
\begin{fact}\label{obs:possibly-bad-upper-bound}
The number of pairs $(U,C)$ with $U\subseteq R$ a set of $\Delta$ offline nodes and $C\subseteq \calC$ of $\varepsilon \Delta$ colors is at most ${n \choose \Delta}\cdot {(1+\sqrt{\eps})\Delta \choose \varepsilon \Delta} \leq n^{\Delta}\cdot 4^{\Delta} \leq n^{2\Delta}$.
\end{fact}

To union bound over all potential \bad pairs (and indirectly to union bound over all possible neighborhoods of $\Delta$ or fewer neighbors that the adversary may select), we prove that any given pair is \bad at any time $t$ with probability at most, say, $n^{-3\Delta}$.
For this, we consider the following processes -- one for each pair $(U,C)$. For now, we fix a particular pair $(U,C)$.

\begin{definition}
    For the above pair $(U,C)$, denote by $Z_i$ the value $\sum_{u \in U} \sum_{c \in C} \frac{\mathds{1}[c\in P(u)]}{|P(u)|} = \sum_{u \in U} \frac{|C\cap P(u)|}{|P(u)|}$ after exactly $i$ edges have picked a color and removed it from their offline node's palette.
    If the overall number of edges revealed by the adversary, $m$, is less than $n\Delta$, we set $Z_i=Z_m$ for all $i=m+1,\dots,Z_{n\Delta}$.
\end{definition}
\begin{Remark}
We will show that $Z_0, Z_1, \ldots,Z_{m}$ forms a martingale in \cref{lemma:martingale-params}. We note that it is in fact an \textbf{exposure} martingale, with $Z_t = \E\left[\sum_{u\in U}\frac{P(u)\cap C}{P(u)} \text{at the end of the algorithm} \ \middle|\ \text{random choices up to time $t$}\right]$, where the expectation is taken over the random choices by the algorithm.
\end{Remark}

Note that the the adversary determines the number of edges, which is at most $n\Delta$ in a $2n$-node graph ($n$ online and $n$ offline nodes), which has maximum degree $\Delta$.

\begin{lemma}\label{lemma:martingale-params}
    The sequence $Z_0,Z_1,\dots,Z_{n\Delta}$ is a martingale (even against adaptive adversaries), with an initial value $Z_0\leq |C|$, step size $|Z_i-Z_{i-1}|\leq \stepsize$ for all $i$ always and observed variance
    $$\sum_{i=1}^m \E\left[(Z_i-Z_{i-1})^2 \mid Z_0,Z_1,\dots, Z_{i-1}\right]\leq \varsize.$$
\end{lemma}
\begin{proof}
    By definition, and since $|U|= \Delta$, we have that $Z_0 = \frac{|U||C|}{(1+\sqrt{\varepsilon}) \Delta} \leq |C|$. We turn to proving the remaining claimed properties of this sequence.
    
    First, if $Z_0,Z_1,\dots,Z_{i-1}$ are assigned values (based on our algorithm's randomness and the adaptively-chosen adversarial node arrivals) $z_0,z_1,\dots,z_{i-1}$ such that the $i^{th}$ edge does not contain an offline node in $U$, then clearly $Z_i=Z_{i-1}$. This implies that for such an assignment:
    \begin{align*}
    \E[Z_i & \mid Z_0=z_0,Z_1=z_1\dots,Z_{i-1}=z_{i-1}] = Z_i,\\
    \E[(Z_i-Z_{i-1})^2 & \mid Z_0=z_0,Z_1=z_1\dots,Z_{i-1}=z_{i-1}] = 0.
    \end{align*}
    This case is therefore consistent with the sequence being a martingale of step size $|Z_i-Z_{i-1}|\leq \stepsize $, and contributes nothing to the observed variance $\sum_{i=1}^m \E\left[(Z_i-Z_{i-1})^2 \mid Z_0,Z_1,\dots, Z_{i-1}\right]$. We may safely focus on the complementary case where $Z_0,Z_1,\dots,Z_{i-1}$ are assigned values $z_0,z_1,\dots,z_{i-1}$, such that the $i^{th}$ edge does contain an offline node $u\in U$. 
    
    Let $P(u)$ and $P_{\text{new}}(u)$ be the palette of $u$ just before and just after the $i^{th}$ edge is processed by \Cref{alg:adaptive-edge-coloring}. So, $|P_{\text{new}}(u)| = |P(u)| - 1$. As conditioning on any realization $Z_0=z_0, Z_1=z_1,\dots,Z_{i-1}=z_{i-1}$ may not provide a closed form for $|P(u)|$ or $|C\cap P(u)|$, we further refine the conditioning, as follows.
    For any non-negative integers $k$ and $\ell \leq k$, let $A(k,\ell)$ be the event that $Z_0=z_0, Z_1=z_1,\dots,Z_{i-1}=z_{i-1}$ and moreover $|P(u)| = k$ and $|P(u) \cap C| = \ell$. (Notice that this event implies that $|P_{\text{new}}(u)| = k-1$.) By definition of $Z_i,Z_{i-1}$, we therefore obtain:
    \begin{align*}
        \E \left[Z_i - Z_{i-1} \mid  A(k,\ell)\right] &
        = \frac{\E[|P_{\text{new}}(u) \cap C|]}{k-1} - \frac{\ell}{k} 
    \end{align*}
    If the color $c$ picked uniformly at random from $P(u)$ by edge $(u,v_t)$ in \Cref{alg:adaptive-edge-coloring} is in $C$, which happens with probability $\frac{\ell}{k}$, then we have that $|P_{\text{new}}(u) \cap C| = \ell - 1$; otherwise, $c \notin C$, and $|P_{\text{new}}(u) \cap C| = \ell$. Overall, we obtain:
    \begin{equation*}
        \E \left[Z_i - Z_{i-1} \mid  A(k,\ell)\right] = \frac{\ell}{k} \cdot \frac{\ell - 1}{k - 1} + \left(1 - \frac{\ell}{k} \right) \cdot \frac{\ell}{k - 1} - \frac{\ell}{k} = 0.
    \end{equation*}
    This equality holds for all the events $A(k,\ell)$ partitioning the event $[Z_0=z_0, Z_1=z_1,\dots, Z_{i-1}=z_{i-1}]$. Hence, by total expectation, we conclude that $Z_0,Z_1,\dots,Z_{n\Delta}$ is a martingale, since for all $z_0,z_1,\dots,z_{i-1}$,
    $$\E[Z_i - Z_{i-1} \mid Z_0=z_0, Z_1=z_1,\dots, Z_{i-1}=z_{i-1}] = 0.$$
    
    Next, similar calculations allow us to bound the contribution of the $i^{th}$ edge containing an offline node $u\in U$ to the observed variance. Again, conditioning on an event $A(k, \ell)$, one can compute:
    \begin{align*}
    \E[(Z_i-Z_{i-1})^2 \mid  A(k,\ell)] & =  \left(\frac{\ell}{k-1} - \frac{\ell}{k}\right)^2 \cdot \left(1-\frac{\ell}{k}\right) + \left(\frac{\ell-1}{k-1}-\frac{\ell}{k}\right)^2 \cdot \frac{\ell}{k}\\
    & \leq \left(\frac{\ell}{k(k-1))}\right)^2 + \left(\frac{\ell}{k(k-1))} - \frac{1}{k-1} \right)^2 \\
    & \leq \frac{2}{(k-1)^2} \\
    & \leq \frac{2}{\varepsilon\cdot \Delta^2}.
    \end{align*}
    Above, the first inequality simply upper bounded both probabilities by one. The penultimate inequality relies on the fact that $\ell = |C\cap P(u)|\leq |P(u)|=k$. Finally, the ultimate inequality relies on $u$ having at most $\Delta-1$ edges revealed before the $i^{th}$ edge (since this edge contains $u$), and hence before our random choices for this edge, $k=|P(u)|=\Delta(1+\sqrt{\varepsilon})-\text{deg}(v)\geq \sqrt{\varepsilon}\Delta+1$. By total expectation, the above implies the following bound on the contribution of the $i^{th}$ edge (if it contains an offline node $u\in U$) to the observed variance: 
    $$\E[(Z_i-Z_{i-1})^2 \mid  Z_0,Z_1,\dots,Z_{i-1}] \leq \frac{2}{\varepsilon\cdot \Delta^2}.$$
    But as the adversary can reveal at most $\Delta^2$ many edges containing an endpoint in $U$, which we recall contains at most $\Delta$ nodes, we obtain the following upper bound on the observed variance:
    $$\E\left[\sum_i (Z_i - Z_{i-1})^2 \;\middle\vert\; Z_0,Z_1,\dots,Z_{i-1}\right] \leq \Delta^2 \cdot \frac{2}{\eps\cdot \Delta^2} = \varsize.$$
    
    It remains to upper bound the step size $|Z_i - Z_{i-1}|$ when the $i^{th}$ edge contains an offline node $u$ in $U$. (We recall that if not, then $Z_i=Z_{i-1}$.) 
    We note that for this $u$ we have $|P_{\text{new}}(u)| = |P(u)| - 1$, that is, there is at precisely one color $c \in P(u)\setminus  P_{\text{new}}(u)$. If this color is in $C$, then we have that $|C\cap P_{\text{new}}(u)|=|C\cap P(u)|-1\leq |P(u)|-1$, and so, 
    \begin{align*}
         |Z_i - Z_{i-1}| \leq \frac{1}{|P(u)|} + |C\cap P_{\text{new}}(u)|\cdot \left(\frac{1}{|P(u)|} - \frac{1}{|P(u)|-1}\right) \leq \frac{2}{|P(u)|}\leq \frac{2}{|P(u)|-1}.
    \end{align*}
    Similarly, if $c$ is not in $C$, then we have that $|C\cap P_{\text{new}}(u)|=|C\cap P(u)|$, and so
    \begin{align*}
         |Z_i - Z_{i-1}| \leq |C\cap P_{\text{new}}(u)|\cdot \left(\frac{1}{|P(u)|} - \frac{1}{|P(u)|-1}\right) \leq \frac{1}{|P(u)|-1}\leq \frac{2}{|P(u)|-1}.
    \end{align*}
    But since before the $i^{th}$ edge is processed $|P(u)|\geq (1+\sqrt{\varepsilon})\Delta-(\Delta-1) = \sqrt{\varepsilon}\Delta+1$, we then have that
    \begin{align*}
        |Z_i - Z_{i-1}| & \leq \frac{2}{|P(u)|-1}\leq \frac{2}{\sqrt{\varepsilon}\Delta},
    \end{align*}
    which finishes the proof.
\end{proof}

We are now ready to prove that for the given pair $(U,C)$ we fixed, the probability the pair is \bad is polynomially smaller than the number of potentially bad pairs, which will momentarily be useful when union bounding over these.

\begin{lemma} \label{lemma:martingale_magic}
For the above pair $(U,C)$ and fixed time $t$, even for an adaptive adversary
\begin{align*}
\Pr[(U,C) \textrm{\ is\ \bad\ at\ time\ } t] \leq n^{-5\Delta+2}.
\end{align*}
\end{lemma}
\begin{proof}
If $i$ is the number of edges revealed until time $t$, then $\sum_{u\in U}\sum_{c\in C} x^{(t)}_{uc} = \sum_{u\in U} \frac{\mathds{1}[c\in P(u)]}{|P(u)|} = Z_i$, with $P(u)$ the palette of $u$ just before time $t$.
Therefore we see that $(U,C)$ is not \bad at time $t$ if $Z_i\leq (1+\varepsilon)\cdot |C|$ for all $i\in [m]$ (note that $i$ is possibly random). We therefore wish to upper bound the deviation of $Z_i$ from $Z_0$ for all $i\in [n\Delta]$, from which we get
\begin{align*}
    \Pr[(U,C) \textrm{\ is \bad\ at\ time } t] \leq \sum_i \Pr[Z_i > (1+\eps)\cdot |C|] \leq \sum_i \Pr[Z_i - Z_0 > \eps\cdot |C|],
\end{align*}
where we recall that $Z_0 < |C|$.
We upper bound the latter terms using Freedman's inequality, applied to each of the martingale states $Z_0,Z_1,\dots,Z_{n\Delta}$, which we recall from \Cref{lemma:martingale-params} has step size at most $A:=\stepsize$ and observed variance at most $\sigma^2:=\varsize$. Since we are interested in a deviation of $\lambda:=\varepsilon\cdot |C| = \eps^2 \Delta$ for $Z_i$ from $Z_0$, from Freedman's inequality (\Cref{thm:freedman_inequality}), we obtain for any $t = 1,\dots,n\Delta$:
    \begin{align*}
        \Pr[(U,C) \textrm{\ is \bad\ at\ time } t] 
        & \leq \sum_i \Pr[Z_i - Z_0 \geq  \lambda] \\
        & \leq \sum_i \exp\left(-\frac{\lambda^2}{2(\sigma^2 + A\lambda/3)}\right) 
        \\ 
        & \leq \sum_i \exp\left(-\frac{\eps^4\Delta^2}{2(\varsize + ((\varepsilon^2 \Delta)\cdot \stepsize)/3)}\right) \\
        & \leq \sum_i \exp(-\eps^5 \Delta^2 / 6) \\
        & \leq n^2\cdot \exp(-\eps^5 \Delta^2 / 6) \\
        & \leq n^2 \cdot \exp(-5\Delta \ln n)\\
        & = n^{-5\Delta+2},
    \end{align*}
    where the last inequality relied on $\eps=2\sqrt[5]{\frac{\ln n}{\Delta}}$ implying $\eps^5 > 30 {\frac{\ln n}{\Delta}}$.
\end{proof}

\begin{corollary} \label{lemma:martingale_cor}
With probability $1-n^{-6}$, at all times, at most~$\eps\cdot \Delta$ many colors are not \good, and this holds even against an adaptive adversary.
\end{corollary}
\begin{proof}
    By \Cref{obs:possibly-bad-upper-bound}, there are at most $n^{2\Delta}$ many possibly bad pairs $(U,C)$ for any fixed time $t$. Therefore, taking a union bound over all pairs and times $t$, we have from \Cref{lemma:martingale_magic} and $\Delta\geq 2$ that 
    $$\Pr[\textrm{some\ pair\ $(U,C)$\ is\ \bad\ for\ some\ time\ $t$}] \leq n^{2\Delta+1}\cdot n^{-5\Delta+2} \leq n^{-6},$$
    where the second inequality used that $\Delta\geq 30 \ln n\geq 3$.
    The corollary then follows by \Cref{obs:no-bad-implies-mostly-good}.
\end{proof}

So far we have established that w.h.p., for every offline node $u$, among the at least $\sqrt{\eps}\cdot \Delta$ colors $c$ still in $u$'s palette $P(u)$ at any given time, at most $\eps\cdot \Delta \ll \sqrt{\eps}\cdot \Delta$ are \bad. The above hints at each node (both offline and online) having each of its edges colored with probability roughly $1-e^{-1}$. More formally, we can show the following.

\begin{lemma}\label{lem:degree-rate-decrease}
Let $H\subseteq G$ be the subgraph of $G$ colored by \Cref{alg:adaptive-edge-coloring}, for $G$ a bipartite graph of maximum degree $\Delta(G)\leq \Delta$, for $\Delta\geq 32\cdot \ln n$.
Then, $\Delta(G\setminus H)\leq (e^{-1}+3\sqrt{\varepsilon})\cdot \Delta$ with probability at least $1-n^{-5}$, even against an adaptive adversary.
\end{lemma}
\begin{proof}
To prove the claim, it suffices to show that all nodes of degree at least $\Delta/e$ in the (adaptive) graph $G$ have at most a $(e^{-1} + 5
\sqrt{\eps})\Delta$ of their edges not colored by \Cref{alg:adaptive-edge-coloring}.
By \Cref{lemma:martingale_cor}, with probability at least $1-n^{-3}$, at any time $t$ at most $\varepsilon \Delta$ colors are not \good. 
We show that subject to this good event, no 
node has more than $(e^{-1}+3\sqrt{\eps})\cdot \Delta$ many edges in $G\setminus H$ with high probability, and so the lemma will follow by union bound over the bad events that some large set of colors is not \good or that some node has degree greater than the above in $G\setminus H$.

Now, consider a time-step $t$. By the above, at most $\varepsilon \Delta$ colors are not \good at time $t$.
On the other hand, each neighbor $u$ of $v_t$ has at least $\sqrt{\eps}{\Delta}$ many of the original palette of $|\calC|=(1+\sqrt{\eps})\Delta$ many colors still in $P(u)$, which decreased as the latter decreased in size by at most $\Delta$ when inspecting the previous edges of $u$.
Therefore, we have that
\begin{equation*}
    \sum_{\textrm{ \good\ }c} x^{(t)}_{uc} 
    \geq 1-\frac{|\{c \mid c \textrm{\ is \ not \good}\}|}{\sqrt{\varepsilon}\Delta} \geq 1-\sqrt{\eps}.
\end{equation*} 
This, together with \Cref{lem:per-edge-coloring-prob}, implies that for all time steps that $u$ belongs to, we have that
\begin{align}\label{eqn:edge-colored-prob-when-good}
\Pr[(u,v_t) \mathrm{\ colored}]\geq \sum_{\textrm{ \good\ }c} x^{(t)}_{uc}\cdot (1-e^{-1}-\eps) \geq (1-\sqrt{\eps})\cdot (1-e^{-1}-\eps) \geq (1-e^{-1}-2\sqrt{\eps}).
\end{align}
Denoting by $Y_i$ an indicator for the $i^{th}$ edge of $u$ not being colored, and denoting by $p=e^{-1}+2\sqrt{\eps}$, we have shown that $\Pr[Y_i=1]\leq p$. Indeed, we have proven the stronger claim that  
$$\Pr[Y_i=1 \mid Y_1,\dots,Y_{i-1}]\leq p.$$
Thus, by standard coupling arguments, and since $u$ has at most $\Delta$ edges chosen by the adversary, we have for $q:=p+\sqrt{\eps}$ that
\begin{align*}
\Pr[\deg_{G\setminus H}(u) \geq q \Delta] & =  \Pr\left[\sum_{i=1}^{\Delta} Y_i \geq q\Delta\right]  \leq \Pr[\Ber(\Delta,p) \geq q \Delta] \leq \exp\left(-\frac{\eps\Delta}{3}\right) \leq n^{-10},
\end{align*}
where the last two inequalities follow from Chenoff-Hoeffding bounds and the lemma's hypothesis whereby $\Delta\geq 32\cdot \ln n$, and so $\eps\Delta = 2\Delta^{4/5}\ln^{1/5}n \geq 32\cdot \ln n$.
Thus, by union bound, we have that with high probability no offline node has more than $q\Delta = (e^{-1}+3\sqrt{\eps})$ many uncolored edges after running \Cref{alg:adaptive-edge-coloring}, provided at most $\eps\cdot \Delta$ colors are not \good at each time $t$.
We turn to prove the same for online nodes.

Consider some online node $v_t$.
By \Cref{eqn:edge-colored-prob-when-good} and linearity of expectation, we have that
\begin{align*}
    \E[\; |\{(u,v_t) \mathrm{\ colored}\}|\;] = \sum_{u\in N(v_t)} \Pr[(u,v_t) \mathrm{\ colored}] \geq \deg(v_t)\cdot (1-e^{-1}-2\sqrt{\eps}).
\end{align*}
Now, recall that by \Cref{lem:CRS}, whenever $c$ is selected by at least one edge $(u,v_t)$ (i.e., $R_c \neq \emptyset$ at time $t$) we assign it to one (unique) such selecting edge of $t$. Therefore, the number of selected colors $C'$ at time $t$ is precisely equal to the number of edges of $v_t$ colored, which is precisely the number of occupied bins in an (asymmetric) balls and bins process: each ball (node) $u$ lands in (selects) bin (color) $c$ independently with probability $x^{(t)}_{uc}$.\footnote{Note that the adversary can select the $x$'s, and the ``independently'' here is with respect to the fresh randomness, given the realized values of $x^{(t)}_{uc}$.}
But by known results in negative correlation \cite{dubhashi1996balls}, the number of occupied bins in a balls and bins process admits Chernoff-Hoeffding type concentration inequalities, and so we obtain the following:
\begin{align*}
    \Pr\Big[ |C'| \leq \E[|C'|] - \sqrt{\varepsilon} \cdot \Delta\Big] \leq \exp\left(-\frac{(\sqrt{\varepsilon}\Delta)^2}{\sum_{u\in N(v_t)} 1^2}\right) = \exp\left(-\frac{\eps \Delta}{1+\sqrt{\eps}}\right) 
    \leq \exp\left(-\sqrt[5]{\frac{\ln n}{ \Delta}}\Delta \right)
    \leq n^{-16},
\end{align*}
where the penultimate inequality relied on $\eps = 2\sqrt[5]{\frac{\ln n}{\Delta}}\leq 1$ and the ultimate inequality relied on $\Delta\geq 32\cdot \ln n$ implying $\Delta^{4/5}\geq 16\cdot \ln^{4/5} n$.
Thus, with probability at least $1-n^{-16}$, online node $v_t$ has a number of selected colors, and hence colored edges, at least $\deg(v_t)\cdot (1-e^{-1}-2\sqrt{\eps})-\sqrt{\eps}\Delta$.
Consequently the number of uncolored edges of $v_t$ is at most $\deg(v_t)\cdot (e^{-1}+2\sqrt{\eps})+\sqrt{\eps}\Delta\leq (e^{-1}+3\sqrt{\eps})\Delta$, with the same probability. Taking union bound over all online nodes and using $n\geq 3$ concludes the proof.
\end{proof}

Given the above, we are now ready to conclude our analysis of \Cref{alg:adaptive-edge-coloring}, by proving \Cref{thm:adaptive-subroutine}, restated below for ease of reference.
\adaptivesubroutine*
\begin{proof}
For each edge $e=(u,v_t)$, \Cref{alg:adaptive-edge-coloring} either leaves the edge uncolored, or it assigns it color $c(e)$, which is removed from $P(u)$ and hence never assigned to an edge of $u$ again. 
Thus, no offline node $u$ has two of its edges assigned the same color.
On the other hand, each online node $v_t$ has at most one edge assigned any fixed color $c$, since for each color we run CRS to select at most one edge $(u,v_t)$ to color $c$.
Thus, \Cref{alg:adaptive-edge-coloring} properly edge colors some subgraph of $G$.
The theorem then follows from \Cref{lem:degree-rate-decrease}.
\end{proof}

\subsection{From Partial Coloring to Edge Coloring}\label{sec:partial-to-full}

\Cref{thm:adaptive-subroutine} implies that, using roughly $\Delta$ colors we can decrease the maximum degree by roughly a factor of $e$. This allows us to pipeline a number of such invocations of \Cref{alg:adaptive-edge-coloring}, using roughly $\Delta\cdot (1+e^{-1}+e^{-2}+\dots)\leq \Delta\cdot \big(\frac{e}{e-1}+o(1)\big)$ many colors to obtain an uncolored subgraph of maximum degree some $o(\Delta)$, which we can color greedily using a further $o(\Delta)$ colors. This brings rise to \Cref{thm:adaptive-algo}, restated here for ease of reference.

\adaptivealgo*
\begin{proof}
     Let $G$ be the input graph (revealed online), and $G_1 \subseteq G$ be the subgraph of $G$ containing the uncolored edges of $G$ after running \Cref{alg:adaptive-edge-coloring} on $G$. 
     By \Cref{thm:adaptive-subroutine}, the maximum degree of $G_1$ is upper bounded by $\Delta_1 := (e^{-1} + 3\sqrt{\varepsilon})\cdot \Delta$ with high probability. 
     We apply \Cref{alg:adaptive-edge-coloring} again, this time on $G_1$, using $\Delta_1$ as (an upper bound on) its maximum degree. This leaves us with another subgraph $G_2 \subseteq G_1$ of uncolored edges, etc. We repeatedly apply this procedure, obtaining a sequence of subgraphs $G \supseteq G_1 \supseteq G_2 \supseteq \dots$, where  every application of \Cref{alg:adaptive-edge-coloring} uses a fresh palette of colors, which by \Cref{thm:adaptive-subroutine} decreases the residual graph's maximum degree by a constant factor (w.h.p.). 
     We stop when the maximum degree is roughly $\Delta^{\alpha} \ln^{1-\alpha} n$ (w.h.p.), for $\alpha=\frac{10}{11}$, and then simply apply the greedy $2$-competitive algorithm on the remaining graph. 
    We note that the resulting algorithm can be implemented online, as showcased in \Cref{alg:pipelining-algorithm}.

     \begin{algorithm}[htb]
    	\caption{Online Edge Coloring}
    	\label{alg:pipelining-algorithm}
    	\begin{algorithmic}[1]
            \State \textbf{Input parameter:} $G$ input graph arriving online, $\Delta$ maximum degree of $G$.

            \State Launch executions $0,\dots,f-1$ of \Cref{alg:adaptive-edge-coloring} with parameters $\Delta_0,\dots,\Delta_{f-1}$ and distinct palettes.
    		\For{\textbf{each} online node $v_t$ on arrival}
                \State $E_t \gets $ edges incident to $v_t$.
    			\For{$i = 1,\dots, f$}
                    \State Provide online arrival of $v_t$ with edges $E_t$ to execution $i$ of \Cref{alg:adaptive-edge-coloring}.
                    \State $E_t \gets E_t \setminus \{\text{edges colored by the above execution}\}$. 
                \EndFor
                
                \State Color remaining set $E_t$ greedily using fixed palette which is distinct from copies of \Cref{alg:adaptive-edge-coloring}.
    		\EndFor
    	\end{algorithmic}	
    \end{algorithm}

     In the following we define the number of applications of \Cref{alg:adaptive-edge-coloring} we use and upper bound the number of colors used and degrees of the resulting subgraphs~$G_i$.
     
     We define the following sequence: $\Delta_0 =\Delta$, and for all $i\geq 0$, we have $\Delta_{i+1} = q\Delta_i = q^{i}\Delta$, for $q := e^{-1} + \lambda$ and $\lambda = 3\sqrt{2}({\ln n/\Delta})^{\alpha/10}$.
     We note that, as $\Delta\geq 10^{11}\ln n$, we have that $\lambda\leq 3\sqrt{2} \cdot 10^{-1} < 1/2$. 
     So, $q=e^{-1}+\lambda <1$, and the sequence $\Delta_i$ strictly decreases as $i$ increases.     
     We run \Cref{alg:adaptive-edge-coloring} on subgraphs $G_i$ for  $i=0,1,\dots,f-1$, where $f := \min \{i \mid \Delta_i \geq \Delta^{\alpha}\ln^{1-\alpha} n\}$, and then run greedy on $G_{f+1}$.

     First, we argue that w.h.p.~against an adaptive adversary, all $\Delta_i$ are upper bounds on the maximum degrees $\Delta(G_i)$ of previously defined subgraphs $G_i$.
     First, by definition of $f$, for all $i\leq f$ we have $\Delta_i \geq \Delta^{\alpha}\ln^{1-\alpha} n = (\Delta/\ln n)^{\alpha}\ln n \geq 32\ln n$, by the hypothesis $\Delta\geq 10^{11}\cdot \ln n$. 
     Next, let $\varepsilon_i := 2\sqrt[5]{\ln n/ \Delta_i}$ be the value of $\eps$ used by \Cref{alg:adaptive-edge-coloring} when running it on $G_i$ with parameter $\Delta_i$. For all $i<f$ we have that 
     $\eps_i \leq 2\sqrt[5]{{\ln n/(\Delta^{\alpha}\ln^{1-\alpha} n)}} = 2({\ln n/\Delta})^{\alpha/5},$ 
     and so $3\sqrt{\eps_i}\leq 3\sqrt{2}({\ln n/\Delta})^{\alpha/10} = \lambda$.
     Let $A_i$ be the event that $\Delta(G_i) \leq \Delta_i$. 
     By \Cref{thm:adaptive-subroutine}, we have that, even against an adaptive adversary,
     \begin{align*}
        \Pr[\overline{A_{i+1}} \mid A_{i-1},A_{i-2},\dots,A_0] & = 
        \Pr[\overline{A_{i+1}} \mid A_{i-1}] \\
        & = \Pr[\Delta(G_{i+1}) > \Delta_{i+1} \mid \Delta(G_{i}) \leq \Delta_i] \\
        & \leq \Pr[\Delta(G_{i+1}) > \Delta_{i}\cdot q  \mid \Delta(G_{i}) \leq \Delta_i] \\
        & \leq \Pr[\Delta(G_{i+1}) > \Delta_{i}\cdot (e^{-1}+3\sqrt{\eps_i}) \mid \Delta(G_{i}) \leq \Delta_i] \\
        & \leq n^{-6}.
     \end{align*}
     Now, let $B_i=\overline{A_i}\wedge\bigwedge_{j<i}A_j$ be the indicator for $i$ being the first index for which $\overline{A_i}$ holds (i.e., for which the subgraph $G_i$ has maximum degree $\Delta(G_i)>\Delta_i$). Then, we have by union bound and the above that
     \begin{align*}
         \Pr\left[\bigvee_i \overline{A_i} \right] = \Pr\left[\bigvee_i B_i \right] \leq \sum_i \Pr\left[B_i \right] \leq f\cdot n^{-6} \leq n^{-5},
     \end{align*}
    where the last inequality follows from the sequence of $\Delta_i$ being strictly decreasing as $i$ increases, and so $f\leq \Delta\leq n$. 
    We conclude that with high probability, all $i$ satisfy $\Delta_i\leq q^i\Delta$, and so in particular
    $\Delta_{f+1}<\Delta^{\alpha}\ln^{1-\alpha}$.
    Consequently, the greedy algorithm uses only $2\Delta^{\alpha}\ln^{1-\alpha} n$ colors when run on $G_{f+1}$. It remains to bound the (deterministic) number of colors used by our successive invocations of \Cref{alg:adaptive-edge-coloring}.
    This number is given by 
     \begin{equation*}
     \sum_{i=0}^{f} \Delta_i (1 + \sqrt{\varepsilon_i}) 
     \leq
     \sum_{i=0}^f q^{i}\Delta (1 + \lambda)
     \le \Delta(1+\lambda)\frac{1}{1-q}
     \end{equation*}
     where the last inequality is a geometric sum  (using that $q<1$).
     Now note that by our choice of $q=e^{-1}+\lambda$, and since $\lambda<1/2$ by the earlier discussion, we have that
     \begin{equation*}
     \frac{1+\lambda}{1-q}
     =
     \frac{1+\lambda}{1-e^{-1}-\lambda}
     =
     \frac{(1+\lambda)e}{e-1-e\lambda}
     = \frac{e}{e-1} + \frac{e(2e-1)\lambda}{(e-1)(e-1-e\lambda)}\leq \frac{e}{e-1} + 20\lambda.
     \end{equation*}

    To conclude, the total number of colors used (both from the successive applications of \Cref{alg:adaptive-edge-coloring} and the final greedy) is, by our choice of $\alpha=\frac{10}{11}$ (satisfying $1-\alpha=\alpha/10$), we have that
    \begin{align*}
    \frac{e}{e-1}\cdot \Delta
    + 20\lambda\Delta + 2\Delta^{\alpha} \ln^{1-\alpha} n 
    & \le \frac{e}{e-1}\cdot \Delta
    + 60\sqrt{2} \Delta^{1-\alpha/10}\ln^{\alpha/10}n + 2\Delta^{\alpha} \ln^{1-\alpha} n 
    \\
    & = \left(\frac{e}{e-1} + 100\sqrt[11]{\ln n/ \Delta} \right)\cdot \Delta,
    \end{align*}
    which ends the proof.
\end{proof}